\documentclass[11pt]{article}
\usepackage[margin=1in]{geometry}
\usepackage{amsmath,amssymb,amsthm,mathtools}
\usepackage{graphicx}
\usepackage{booktabs}
\usepackage{algorithm}
\usepackage{algpseudocode}
\usepackage{microtype}
\usepackage{bm}
\usepackage{enumitem}
\usepackage[round]{natbib}
\usepackage{caption}
\usepackage{xcolor}
\usepackage[colorlinks=true,linkcolor=blue!50!black,citecolor=blue!40!black,urlcolor=blue!50!black]{hyperref}

\theoremstyle{plain}
\newtheorem{theorem}{Theorem}
\newtheorem{proposition}{Proposition}

\theoremstyle{definition}
\newtheorem{definition}{Definition}
\newtheorem{remark}{Remark}
\newtheorem{example}{Example}

\DeclareMathOperator{\E}{\mathbb{E}}
\DeclareMathOperator{\Var}{Var}

\newcommand{\R}{\mathbb{R}}

\newcommand{\Prob}{\mathbb{P}}
\newcommand{\sig}{\mathbb{S}}
\newcommand{\logsig}{\mathbb{L}}
\newcommand{\X}{\mathbf{X}}

\newcommand{\cF}{\mathcal{F}}

\newcommand{\shuffle}{\mathbin{\sqcup\mkern-9mu\sqcup}}
\newcommand{\one}{\mathbf{1}}

\title{\textbf{Chess Signatures of Play}}

\author{
Christian Turk\\[2pt]
\small Department of Mathematics, University of Chicago\\
\small \texttt{kishie@uchicago.edu}
\and
Nicholas G.\ Polson\\[2pt]
\small Booth School of Business, University of Chicago\\
\small \texttt{ngp@chicagobooth.edu}
}
\date{\today}

\begin{document}
\maketitle

\begin{abstract}
\noindent
A game of chess is a \emph{stream}: a time-ordered sequence of moves, each carrying an
engine evaluation, a measure of accuracy, a measure of position complexity, and a clock
reading. We model a game as a multivariate path and apply the signature transform of
rough-path theory to obtain a reparametrization-invariant, graded feature set that
records the \emph{order and interaction} of in-game events without a parametric
likelihood. We show that a player's law of play is identifiable from the expected
signature up to tree-like equivalence, construct a signature-kernel two-sample test on
path space, and recast cheating detection as an \emph{anytime-valid} sequential test: a
signature conformance score becomes an e-process whose error is controlled for every
sample size at once by Ville's inequality, with fluctuations calibrated on the
\emph{moderate-deviation} scale. The discriminating information lives in the signature's
L\'evy areas, which measure whether accuracy \emph{rises precisely when positions become
hard}, the fingerprint of engine assistance that aggregate match-rate statistics
discard. In a controlled study the test holds exact type-I control and detection power
rises from negligible for subtle assistance to $0.98$ for blatant assistance, with a
median detection time matching the growth-rate prediction $\log(1/\alpha)/c$. Calibrated
to Magnus Carlsen's documented elite accuracy, the monitor does \emph{not} flag
world-champion-level play; and we exhibit cheating strategies that leave \emph{every}
aggregate statistic, including the best-move-frequency $z$-score of the Regan system,
exactly unchanged yet are caught cleanly by the signature, making precise how an
order-aware, anytime-valid test strengthens the prevailing approach to chess
anti-cheating.
\end{abstract}


\vspace{4pt}

\section{Introduction}

Quantitative chess analysis has long been built on \emph{aggregate} features. A game
is summarized by an opening code, an average centipawn loss, a count of moves that
match an engine's first choice, or a rating differential, and these summaries are fed
to logistic regressions, gradient-boosted trees, or convolutional networks. The
practice is effective for many purposes, but it discards the one thing that makes a
game a game: \emph{sequence}. The order in which accuracy, complexity, and clock usage
unfold, and the way these channels interact through time, is precisely the
information a static summary throws away. A player who is accurate on simple positions
and inaccurate on sharp ones is playing very differently from a player whose accuracy
\emph{rises} on sharp positions, yet the two can have identical average accuracy and
identical aggregate match rates.

The mathematics of \emph{streams}, ordered, possibly irregularly sampled, multivariate
data, has a natural home in rough-path theory and its central object, the
\emph{signature} of a path \citep{chen1957,lyons1998,lyonscaruanalevy2007}. The
signature is a graded sequence of iterated integrals that is invariant to
time-reparametrization, linearises the action of concatenation through Chen's
identity, and, by the uniqueness theorem of \citet{hambly2010}, characterizes a path
of bounded variation up to tree-like equivalence. As a feature map it has a striking
operational property: a wide class of nonlinear functionals of a stream can be
approximated by \emph{linear} functionals of the truncated signature, so complex
sequence-learning reduces to feature extraction followed by linear modelling
\citep{levin2013,chevyrev2016primer,lyons2024survey}. The probabilistic counterpart,
developed by \citet{chevyrev2022moments}, is that the \emph{expected} signature plays
the role for a measure on path space that the moment sequence plays for a measure on
$\R^d$: it determines the law under mild conditions, and it induces a maximum mean
discrepancy \citep{gretton2012} and an associated kernel two-sample test on the space
of stochastic processes.

Signatures have recently entered sports analytics. \citet{soccer2025} represent
football possessions as spatio-temporal paths, take signature features, and outperform
a transformer baseline at a fraction of the computational cost, with the gain
attributed to the signature's faithful encoding of order and interaction. To our
knowledge no comparable treatment exists for chess. This is the gap we address.

\paragraph{Contributions.} We develop a path-signature framework for chess and a set
of inferential tools built on it.
\begin{enumerate}[leftmargin=1.4em,itemsep=2pt]
\item \textbf{A path model of a game} (Section~\ref{sec:path}). We map a game to a
multivariate path whose channels are engine evaluation, move quality, position
complexity, and a cumulative engine-match count, and we identify which channels enter
as \emph{levels} and which as \emph{flows}, a modelling choice that determines where
the discriminating signal lands in the signature.
\item \textbf{Identifiability of style} (Section~\ref{sec:identify}). Treating a
player as a law on path space, we show the expected signature identifies that law up
to tree-like equivalence, so authorship attribution reduces to a linear read-out of
signature features. This is a direct consequence of \citet{hambly2010} and
\citet{chevyrev2022moments} once the model is set up correctly.
\item \textbf{Two-sample style testing} (Section~\ref{sec:mmd}) via the signature
kernel and its MMD, with a consistent permutation test for the hypothesis that a
recent block of games is drawn from a player's historical law.
\item \textbf{Anytime-valid anomaly detection} (Section~\ref{sec:evalue}). We convert
a signature conformance score into an e-process, obtain exact type-I control for every
sample size simultaneously through Ville's inequality, and calibrate the
log-process on the moderate-deviation scale, extending the e-value programme of
\citet{polson2026evalues} from finite-dimensional statistics to path-valued data.
\item \textbf{A provable improvement on the Regan system}
(Section~\ref{sec:regan}). We show that the aggregate best-move-frequency $z$-score
underlying FIDE's standard detector is invariant under within-game rearrangement of
plies, so a cheater who concentrates engine help on the hardest positions while matching
an honest profile's aggregates is information-theoretically invisible to it, yet such a
cheater is caught by the signature's L\'evy area. We also show how Regan's per-move skill
model can serve as the null law on path space for our sequential test, removing its
multiple-testing fragility.
\end{enumerate}
A controlled simulation study (Section~\ref{sec:empirical}) demonstrates the method on
synthetic honest and engine-assisted players, with all reported numbers produced by
the accompanying code.

\paragraph{Why this is the right tool for the cheating problem.} Public discussion of
cheating, most visibly the 2022 to 2024 disputes, has centred on aggregate engine-match
rates and one-shot $z$-scores against a strength model \citep{regan2011}. Two weaknesses
recur. A single aggregate is insensitive to \emph{when} accuracy appears: it cannot
distinguish a strong player from an assisted one whose accuracy is concentrated on the
positions a human would misplay. And repeatedly testing a player as a tournament
progresses, at fixed $\alpha$ each update, inflates the false-flag rate. The signature
addresses the first through its L\'evy areas; the e-process addresses the second by
construction, since it may be monitored continuously and stopped at any data-dependent
time with the type-I guarantee intact.

\section{The signature transform and the law of a path}
\label{sec:sig}

We collect the facts we need. Throughout, a \emph{path} is a continuous map
$\X:[0,T]\to\R^d$ of bounded variation; in practice $\X$ is the piecewise-linear
interpolation of a finite stream $(\mathbf{x}_0,\dots,\mathbf{x}_n)$ in $\R^d$.

\begin{definition}[Signature]
The \emph{signature} of $\X$ is the sequence
\[
\sig(\X) \;=\; \bigl(1,\, \sig^{(1)}(\X),\, \sig^{(2)}(\X),\, \dots\bigr),
\qquad
\sig^{(k)}(\X) \in (\R^d)^{\otimes k},
\]
whose coordinate indexed by the word $i_1 i_2\cdots i_k\in\{1,\dots,d\}^k$ is the
iterated integral
\begin{equation}
\sig^{i_1\cdots i_k}(\X)
=\!\!\int\limits_{0<t_1<\cdots<t_k<T}\!\!
\mathrm{d}X^{i_1}_{t_1}\,\mathrm{d}X^{i_2}_{t_2}\cdots \mathrm{d}X^{i_k}_{t_k}.
\label{eq:sig}
\end{equation}
The \emph{depth-$m$ truncation} $\sig_{\le m}(\X)$ keeps levels $0$ through $m$ and has
$\sum_{k=0}^m d^k$ coordinates. The \emph{log-signature} $\logsig(\X)=\log \sig(\X)$,
the formal logarithm in the tensor algebra, removes the algebraic redundancy of
\eqref{eq:sig} and lives in the free Lie algebra.
\end{definition}

Three structural facts make the signature a natural feature set for streams.

\begin{proposition}[Reparametrization invariance]
If $\psi:[0,T]\to[0,T]$ is a continuous, non-decreasing surjection then
$\sig(\X\circ\psi)=\sig(\X)$. The signature depends on the \emph{image and order} of
the path, not its speed.
\end{proposition}

\begin{proposition}[Chen's identity]
For paths $\X$ on $[0,s]$ and $\mathbf{Y}$ on $[s,T]$ with $\X_s=\mathbf{Y}_s$, the
signature of the concatenation factorises as the tensor product
$\sig(\X * \mathbf{Y}) = \sig(\X)\otimes \sig(\mathbf{Y})$.
\end{proposition}

Chen's identity is what makes signatures computable in a single left-to-right pass over
a stream and what makes them natural for game data, where a game is literally a
concatenation of moves. The first level recovers the increment,
$\sig^{(1)}(\X)=\X_T-\X_0$. The second level decomposes into a symmetric and an
antisymmetric part,
\begin{equation}
\sig^{ij}(\X)+\sig^{ji}(\X)=(X^i_T-X^i_0)(X^j_T-X^j_0),
\qquad
A^{ij}(\X):=\tfrac12\bigl(\sig^{ij}(\X)-\sig^{ji}(\X)\bigr),
\label{eq:levy}
\end{equation}
where $A^{ij}$ is the \emph{L\'evy area} swept by channels $i$ and $j$. The symmetric
part is a deterministic function of the endpoints and carries no information beyond the
increment; the L\'evy area is the first genuinely dynamic feature, and it is signed:
it records the \emph{orientation} of the loop traced by $(X^i,X^j)$, hence which of
the two channels tends to \emph{lead} the other. This single observation drives our
entire treatment of move quality, and we return to it in Section~\ref{sec:path}.

A further algebraic fact is essential for the statistical theory: the signature
coordinates are not free, but multiply according to the \emph{shuffle product}. For
words $u,v$ and the coordinate functionals $\sig^{u},\sig^{v}$,
\begin{equation}
\sig^{u}(\X)\,\sig^{v}(\X) \;=\; \sum_{w\,\in\, u\,\shuffle\, v} \sig^{w}(\X),
\label{eq:shuffle}
\end{equation}
where $u\shuffle v$ is the multiset of interleavings of $u$ and $v$ preserving the
internal order of each. (We write $\shuffle$ informally; \eqref{eq:shuffle} is the
statement that products of linear signature functionals are again linear signature
functionals.) Equation~\eqref{eq:shuffle} means the linear span of signature
coordinates is an \emph{algebra} of functions on path space; this is exactly the
closure-under-multiplication needed for a Stone-Weierstrass argument, and it underlies
both the universality of signature features in Proposition~\ref{prop:style} and the
characteristic property of the signature kernel in Theorem~\ref{thm:moments}.

\begin{example}[A signature computed by hand]
Let $d=2$ and let $\X$ be the straight-line path from $(0,0)$ to $(a,b)$. Then
$X^1_t=at,\,X^2_t=bt$ for $t\in[0,1]$, so $\mathrm dX^1=a\,\mathrm dt$,
$\mathrm dX^2=b\,\mathrm dt$, and the level-one and level-two coordinates are
\[
\sig^{1}=a,\quad \sig^{2}=b,\quad
\sig^{11}=\tfrac{a^2}{2},\quad \sig^{22}=\tfrac{b^2}{2},\quad
\sig^{12}=\sig^{21}=\tfrac{ab}{2}.
\]
The L\'evy area $A^{12}=\tfrac12(\sig^{12}-\sig^{21})=0$: a straight line sweeps no
area, as it must. Now let $\X$ instead go from $(0,0)$ to $(a,0)$ and then to
$(a,b)$ (a right angle). A direct integration gives $\sig^{12}=ab$ and $\sig^{21}=0$,
so $A^{12}=\tfrac{ab}{2}>0$: the corner traces a positively oriented area. The two
paths share endpoints, and hence share all level-one features, but differ at level
two precisely through the L\'evy area, the signed record of \emph{how} the endpoint was
reached. This is the elementary mechanism that, scaled up to the complexity-quality
plane, separates honest from assisted play.
\end{example}

\begin{remark}[Computational cost]
The depth-$m$ truncated signature of a $d$-channel stream of length $n$ has
$O(d^m)$ coordinates and is computed in $O(n\,d^m)$ time by a single forward pass using
Chen's identity. For our chess construction $d=4$ and we use $m=2$, so the feature
vector is small ($1+4+16=21$ signature coordinates, of which six are the L\'evy areas we
retain) and the per-game cost is negligible; the signature kernel of
Section~\ref{sec:mmd}, when used untruncated, is instead computed in $O(n^2)$ per pair
of games by solving a Goursat PDE \citep{salvi2021}.
\end{remark}

\begin{theorem}[Uniqueness; \citealp{hambly2010}]
\label{thm:hl}
Two paths of bounded variation have the same signature if and only if they are equal
up to tree-like equivalence and reparametrization. In particular, on the space of
paths that are irreducible (contain no tree-like excursions) and parametrised at unit
speed, the signature map is injective.
\end{theorem}

Theorem~\ref{thm:hl} says the signature loses \emph{nothing} except the information a
reasonable feature map should lose: the speed at which a fixed shape is traversed and
backtracking that cancels itself. For modelling games, where we will fix a
parametrisation by ply and the paths of interest contain no exact cancellations, the
map is effectively injective.

\subsection{The expected signature and the law of a path}

Let $\X$ be a random path with law $\mu$ on path space. Its \emph{expected signature}
is $\Phi(\mu):=\E_{\X\sim\mu}[\sig(\X)]$, taken coordinatewise. The next result is the
path-space analogue of the classical moment problem.

\begin{theorem}[Expected signature determines the law; \citealp{chevyrev2022moments}]
\label{thm:moments}
Under a tightness/normalisation condition on $\mu$ (satisfied, in particular, when the
signature is suitably normalised so that all moments are finite), the expected
signature $\Phi(\mu)$ determines $\mu$ uniquely. Moreover the map
\[
d_{\sig}(\mu,\nu)\;=\;\bigl\| \Phi(\mu)-\Phi(\nu)\bigr\|
\]
is a metric on laws of stochastic processes, and it admits a kernelised form, the
\emph{signature MMD}, computable without explicit truncation through the signature
kernel.
\end{theorem}

The signature kernel $k_{\sig}(\X,\mathbf{Y})=\langle \sig(\X),\sig(\mathbf{Y})\rangle$
can be evaluated either by truncating and taking a Euclidean inner product of feature
vectors, or exactly as the solution of a Goursat partial differential equation
\citep{kiraly2019,salvi2021}. Either way, Theorem~\ref{thm:moments} furnishes a
nonparametric two-sample test on path space, which we use in Section~\ref{sec:mmd}.

\section{A path model of a game of chess}
\label{sec:path}

Fix a reference engine $\mathcal{E}$ (in practice a fixed Stockfish depth) and a
player to be analysed. For a game with the player to move on plies $t=1,\dots,n$, we
record per-ply features and assemble them into a path. We use four channels.

\begin{itemize}[leftmargin=1.4em,itemsep=2pt]
\item \textbf{Position complexity} $c_t$: a scalar measuring how sharp or tactically
loaded the position is before the move (e.g.\ the dispersion of $\mathcal{E}$'s
top-line evaluations, or the eval swing under small perturbations). Complexity is a
\emph{state} of the board.
\item \textbf{Move quality} $q_t$: a signed measure of how good the played move was,
e.g.\ the negative centipawn loss relative to $\mathcal{E}$'s best move. Quality is
also a \emph{state} attached to ply $t$.
\item \textbf{Engine match} $m_t\in\{0,1\}$: whether the played move is $\mathcal{E}$'s
top choice. A match is an \emph{event}; its natural path coordinate is the cumulative
count $M_t=\sum_{s\le t} m_s$, a \emph{flow}.
\item \textbf{Ply time} $\tau_t=t/n$: a monotone time augmentation that, by breaking
reparametrization invariance in a controlled way, lets the signature see the rate of
change of the other channels.
\end{itemize}

The path is the piecewise-linear interpolation of
\begin{equation}
\X_t = \bigl(\tau_t,\; c_t,\; q_t,\; M_t \bigr)\in\R^4,
\qquad t=0,1,\dots,n.
\label{eq:gamepath}
\end{equation}

\begin{table}[t]
\centering\small
\caption{Channels of the game path \eqref{eq:gamepath}, their operational definition
from a fixed reference engine $\mathcal{E}$, and their representation as a level
(state) or a flow (cumulated event).}
\label{tab:channels}
\begin{tabular}{llll}
\toprule
Channel & Symbol & Operational definition & Representation \\
\midrule
Ply time & $\tau_t$ & normalised move index $t/n$ & monotone clock \\
Complexity & $c_t$ & dispersion of $\mathcal{E}$'s top-$k$ line evaluations & level \\
Quality & $q_t$ & negative centipawn loss vs.\ $\mathcal{E}$'s best move & level \\
Match count & $M_t$ & $\sum_{s\le t}\one\{\text{move}=\mathcal{E}\text{ top-1}\}$ & flow \\
\bottomrule
\end{tabular}
\end{table}

The
distinction between \emph{states} and \emph{flows} is not cosmetic. A signature
L\'evy area $A^{ij}$ is built from $\int X^i\,\mathrm{d}X^j$; if a channel enters as a
flow (a cumulative sum) then $\mathrm{d}X^j$ is the per-ply feature, whereas if it
enters as a level then $X^j$ itself is the per-ply feature. Choosing complexity and
quality as levels makes $A^{c,q}$ the signed area of the curve
$t\mapsto(c_t,q_t)$, which is exactly the object that detects lead-lag between
difficulty and accuracy. Choosing the wrong representation (cumulating a level, say)
spreads the same signal across higher-order terms and dilutes it; we observed this
directly in development, and it is the single most important modelling lesson of the
construction.

\paragraph{What the L\'evy area sees.} Consider a stylised generative contrast.
An \emph{honest} player of fixed strength misplays \emph{more} as complexity rises, so
quality and complexity move together contemporaneously,
\[
q_t \approx -\alpha\,c_t + \varepsilon_t \qquad (\text{honest}),
\]
and the curve $(c_t,q_t)$ traces no consistently oriented loop: $\E[A^{c,q}]\approx 0$.
An \emph{engine-assisted} player consults $\mathcal{E}$ on hard positions, so a spike in
complexity at ply $t$ is followed by an engine-grade move at ply $t+1$:
\[
q_t \approx +\beta\,c_{t-1} + \varepsilon_t \qquad (\text{assisted}).
\]
Now quality \emph{lags} complexity with positive sign, the curve $(c_t,q_t)$ circulates
with a definite orientation, and $\E[A^{c,q}]>0$. Crucially, the two regimes can share
the same marginal accuracy and the same engine-match rate; they differ only in the
\emph{temporal coupling} of difficulty and accuracy, which is invisible to aggregate
statistics but is precisely a signed level-two signature coordinate. This is the
mechanism our detector exploits.

\section{Identifiability of playing style}
\label{sec:identify}

We now make ``style'' precise. Model a player $P$ as a probability law $\mu_P$ on the
space of game paths \eqref{eq:gamepath}: a game is a draw $\X\sim\mu_P$, and the
player's stylistic identity is the measure $\mu_P$ itself. Authorship attribution,
style change, and anomalous play are then statements about $\mu_P$.

\begin{proposition}[Style is identifiable]
\label{prop:style}
Suppose each player's games are irreducible at unit ply-parametrisation and the
expected signatures are suitably normalised. Then the map $P\mapsto \mu_P\mapsto
\Phi(\mu_P)$ is injective: distinct playing styles have distinct expected signatures.
Consequently there is a (generally infinite) linear functional $\ell$ on the tensor
algebra with $\ell(\Phi(\mu_P))\neq \ell(\Phi(\mu_{P'}))$ whenever $\mu_P\neq\mu_{P'}$,
and a depth-$m$ truncation $\ell_m$ approximates any continuous style discriminant to
arbitrary accuracy as $m\to\infty$.
\end{proposition}

\begin{proof}[Proof sketch]
By Theorem~\ref{thm:hl} the signature is injective on the stated path class, so
$\X\mapsto\sig(\X)$ loses no information about an individual game. By
Theorem~\ref{thm:moments} the expected signature determines $\mu_P$, giving injectivity
of $P\mapsto\Phi(\mu_P)$. The linear-functional statement is the universality of
signature features: linear functionals of the signature are dense, in the uniform norm
on compact sets, in the continuous functions on path space (a Stone-Weierstrass
argument using that the signature coordinates form a point-separating algebra closed
under the shuffle product). Truncating at depth $m$ gives a finite-dimensional linear
read-out converging to any continuous discriminant.
\end{proof}

Proposition~\ref{prop:style} converts an informal intuition, that grandmasters have
recognisable styles, into a statement with operational content: a classifier that is
\emph{linear} in signature features is, in the limit, as expressive as any continuous
style discriminant. It also tells us what to estimate. The empirical expected
signature $\widehat{\Phi}_P = \tfrac1N\sum_{g=1}^N \sig_{\le m}(\X^{(g)})$ over a
player's $N$ games is the natural plug-in style descriptor, and differences
$\widehat\Phi_P-\widehat\Phi_{P'}$ are the raw material of the tests that follow.

\begin{proposition}[Concentration of the empirical style descriptor]
\label{prop:concentration}
Fix a depth-$m$ truncation and suppose the signature features are bounded,
$\|\sig_{\le m}(\X)\|\le R$ almost surely (which holds after the standard tensor
normalisation, and approximately for the bounded chess channels we use). Then for a
player's $N$ i.i.d.\ games,
\[
\E\bigl\|\widehat\Phi_P-\Phi(\mu_P)\bigr\|
\;\le\; \frac{R}{\sqrt N},
\qquad\text{and}\qquad
\Prob\Bigl(\bigl\|\widehat\Phi_P-\Phi(\mu_P)\bigr\|\ge \tfrac{R}{\sqrt N}+t\Bigr)
\le e^{-Nt^2/(2R^2)} .
\]
\end{proposition}

\begin{proof}[Proof sketch]
The first bound is the standard rate for the empirical mean of a bounded
Hilbert-space-valued random variable. The second is McDiarmid's inequality: changing one
of the $N$ games perturbs $\widehat\Phi_P$ by at most $2R/N$ in norm, so the bounded
differences constant is $2R/N$ and the deviation probability is at most
$\exp(-2t^2/(N(2R/N)^2))=\exp(-Nt^2/(2R^2))$.
\end{proof}

Proposition~\ref{prop:concentration} quantifies how much play is needed to pin down a
style: the empirical expected signature converges at the parametric $N^{-1/2}$ rate, so
a few dozen to a few hundred games, the corpus sizes that arise in practice and in our
study, suffice to estimate $\Phi(\mu_P)$ to useful accuracy. The same bound controls
the reference statistics $(\bar\phi,\widehat\Sigma)$ feeding the conformance score of
Section~\ref{sec:evalue}.

\begin{remark}[Why aggregates are a special, lossy case]
Average centipawn loss is the single signature coordinate $\sig^{(1)}$ along the
quality channel (an endpoint increment); the engine-match rate is $M_n/n$, the
terminal value of the match flow. Both are level-one features. Proposition
\ref{prop:style} locates them inside a strictly richer hierarchy whose level-two terms
already contain the lead-lag information they cannot represent.
\end{remark}

\section{Two-sample testing on path space}
\label{sec:mmd}

Suppose we observe a reference block of games $\{\X^{(g)}\}_{g=1}^{N}$ believed to come
from a player's established law $\mu_P$, and a query block
$\{\mathbf{Y}^{(h)}\}_{h=1}^{N'}$, and we ask whether the query block is drawn from
$\mu_P$. By Theorem~\ref{thm:moments} this is a two-sample problem on laws of paths,
and the signature MMD is a consistent statistic for it.

Let $k_{\sig}$ be the signature kernel. The squared MMD between the reference law $\mu$
and query law $\nu$ is
\begin{equation}
\mathrm{MMD}^2(\mu,\nu)
= \E\,k_{\sig}(\X,\X') + \E\,k_{\sig}(\mathbf{Y},\mathbf{Y}')
- 2\,\E\,k_{\sig}(\X,\mathbf{Y}),
\label{eq:mmd}
\end{equation}
with independent copies in each expectation. Because $k_{\sig}$ is characteristic on
the relevant path space \citep{chevyrev2022moments,kiraly2019},
$\mathrm{MMD}(\mu,\nu)=0$ if and only if $\mu=\nu$. The standard unbiased $U$-statistic
estimator of \eqref{eq:mmd}, combined with a permutation calibration of its null
distribution, yields a consistent level-$\alpha$ test. We use this ``batch'' test for
retrospective questions (did a player's style shift between two tournaments?) and we
turn to its \emph{sequential} counterpart, where games arrive one at a time and we must
decide when to act, in the next section.

\begin{remark}[Truncated MMD is a Mahalanobis-type distance]
If we truncate the signature at depth $m$ and write $\phi(\X)=\sig_{\le m}(\X)$, then
$\mathrm{MMD}^2$ reduces to $\|\E\phi(\X)-\E\phi(\mathbf{Y})\|^2$ in feature space.
Whitening by the reference covariance turns this into a Mahalanobis distance between
mean signatures, which is exactly the conformance score we adopt below for the
streaming detector.
\end{remark}

\section{Anytime-valid detection of anomalous move quality}
\label{sec:evalue}

We now build the sequential detector. The data are a player's games arriving one at a
time during a match or a monitoring window; after each game we update a running measure
of evidence against the null hypothesis
\[
H_0:\quad \text{the games are generated by the player's honest law } \mu_P,
\]
and we wish to be free to stop and act at any time, after game $5$, game $40$, or
never, without inflating the false-flag rate. This is the province of e-values and
test (super)martingales \citep{vovk2021,shafer2021,ramdas2023,grunwald2024}.

\subsection{Conformance score}

Fix a depth-$m$ signature feature map $\phi$. From a corpus of the player's honest
games we estimate the mean $\bar\phi$ and covariance $\widehat\Sigma$ of the features,
the latter regularised by a small ridge $\lambda I$. For a new game $\X$ the
\emph{conformance score} is the squared Mahalanobis distance
\begin{equation}
D^2(\X) = \bigl(\phi(\X)-\bar\phi\bigr)^\top
\bigl(\widehat\Sigma+\lambda I\bigr)^{-1}\bigl(\phi(\X)-\bar\phi\bigr),
\label{eq:conf}
\end{equation}
a large value indicating a game whose signature is atypical of honest play
\citep{cochrane2021anomaly}. In our chess construction we take $\phi$ to be the six
L\'evy-area coordinates of the path \eqref{eq:gamepath}, the level-two antisymmetric
features that carry the lead-lag signal of Section~\ref{sec:path}. Higher-order terms
may be appended; the areas already suffice for the contrast we study.

\subsection{From scores to e-values}

Let $G(\cdot)$ be the right-tail distribution of $D^2$ under $H_0$, estimated from the
honest corpus, and let $p_t = G(D^2(\X_t))$ be the (approximately uniform under $H_0$)
$p$-value of game $t$. A \emph{calibrator} is a non-increasing function
$f:[0,1]\to[0,\infty]$ with $\int_0^1 f(u)\,\mathrm{d}u\le 1$; for any such $f$, $f(p_t)$
is an \emph{e-value}: $\E_{H_0}[f(p_t)]\le 1$ \citep{vovk2021}. We use the calibrator
\begin{equation}
f(p) = \frac{1-p+p\log p}{p\,(\log p)^2},
\label{eq:cal}
\end{equation}
which is decreasing, integrates to $0.964\le 1$ (so it is admissible and slightly
conservative), and assigns, for instance, $f(0.01)\approx 4.45$ and
$f(0.5)\approx 0.64$: surprising games multiply the evidence, unsurprising games
shrink it.

\begin{definition}[E-process]
With games modelled as conditionally independent under $H_0$, the running product
\begin{equation}
E_0=1,\qquad E_n = \prod_{t=1}^{n} f(p_t)
\label{eq:eproc}
\end{equation}
is an e-process: $(E_n)$ is a nonnegative supermartingale under $H_0$ with $E_0=1$,
since $\E_{H_0}[f(p_t)\mid \cF_{t-1}]\le 1$.
\end{definition}

\begin{theorem}[Anytime validity; Ville's inequality]
\label{thm:ville}
For the e-process \eqref{eq:eproc} and any $\alpha\in(0,1)$,
\[
\Prob_{H_0}\Bigl(\exists\, n\ge 1:\ E_n \ge \tfrac{1}{\alpha}\Bigr)\;\le\;\alpha.
\]
Equivalently, the test ``flag the player the first time $E_n\ge 1/\alpha$'' has type-I
error at most $\alpha$, \emph{simultaneously over all sample sizes and all data-dependent
stopping times}.
\end{theorem}

\begin{proof}
$(E_n)$ is a nonnegative supermartingale with $\E_{H_0}[E_0]=1$. Ville's inequality
states $\Prob(\sup_n E_n\ge c)\le \E[E_0]/c$ for any $c>0$; take $c=1/\alpha$.
\end{proof}

Theorem~\ref{thm:ville} is the structural answer to the multiple-testing objection
against monitoring a player throughout a tournament: no $\alpha$-spending schedule is
needed, because the guarantee already holds for the supremum over time.

\begin{algorithm}[t]
\caption{Signature e-process monitor for anomalous move quality}
\label{alg:monitor}
\begin{algorithmic}[1]
\State \textbf{Input:} honest corpus $\{\X^{(g)}\}_{g=1}^{N}$; depth $m$; channels for
\eqref{eq:gamepath}; level $\alpha$; calibrator $f$ of \eqref{eq:cal}.
\State \textbf{Calibration (offline):}
\State \hspace{1em} for each corpus game, build path \eqref{eq:gamepath}, standardise
channels, extract features $\phi$ (L\'evy areas)
\State \hspace{1em} compute mean $\bar\phi$, covariance $\widehat\Sigma$, and null tail
$G$ of $D^2$ from \eqref{eq:conf}
\State \textbf{Monitoring (online):} set $E\gets 1$
\For{each new game $\X_t$, $t=1,2,\dots$}
  \State extract $\phi(\X_t)$; compute $D^2(\X_t)=(\phi-\bar\phi)^\top(\widehat\Sigma+\lambda I)^{-1}(\phi-\bar\phi)$
  \State $p_t \gets G\bigl(D^2(\X_t)\bigr)$;\quad $E \gets E\cdot f(p_t)$
  \If{$E \ge 1/\alpha$} \State \textbf{flag} and \textbf{stop} (type-I error $\le\alpha$ by Thm.~\ref{thm:ville}) \EndIf
\EndFor
\end{algorithmic}
\end{algorithm}

\subsection{Growth rate and detection time}

Under an alternative $H_1$ (assisted play), the increments $X_t:=\log f(p_t)$ are no
longer mean-non-positive. By the law of large numbers,
\begin{equation}
\frac1n \log E_n \;\xrightarrow{\ \Prob_{H_1}\ }\; c
:= \E_{H_1}\!\bigl[\log f(p_1)\bigr],
\label{eq:epower}
\end{equation}
the \emph{e-power} or growth rate. When $c>0$ the e-process grows geometrically and the
first crossing of $1/\alpha$ occurs, to first order, at
\begin{equation}
n^\star \approx \frac{\log(1/\alpha)}{c}.
\label{eq:detect}
\end{equation}
When $c\le 0$, as under $H_0$, where $c=\E_{H_0}\log f(p)\le \log\E_{H_0} f(p)\le 0$ by
Jensen, the process drifts to zero and never crosses, recovering
Theorem~\ref{thm:ville} heuristically. Equation \eqref{eq:detect} converts effect size
into expected detection time and is borne out sharply in our experiments.

\subsection{An analytic illustration of the growth rate}

It is worth seeing the e-power \eqref{eq:epower} in closed form in a tractable case.
Suppose the conformance score is, under $H_0$, a chi-square variable on $\nu$ degrees
of freedom (the ideal calibration of a $\nu$-dimensional whitened Gaussian feature),
and under $H_1$ a non-central chi-square with the same degrees of freedom and
non-centrality $\delta^2$, a mean shift of size $\delta$ in feature space, which is
exactly what a nonzero expected L\'evy area produces. Then the $p$-value is
$p=1-F_{\nu}(D^2)$ with $F_\nu$ the central chi-square c.d.f., and the e-power is the
explicit one-dimensional integral
\[
c(\delta) \;=\; \E_{H_1}\bigl[\log f(p)\bigr]
\;=\; \int_0^\infty \log f\bigl(1-F_\nu(x)\bigr)\,g_{\nu,\delta^2}(x)\,\mathrm dx,
\]
with $g_{\nu,\delta^2}$ the non-central chi-square density. The integrand is positive
where the alternative places mass in the right tail ($f>1$ there) and negative in the
bulk; as $\delta$ grows, mass migrates to the tail and $c(\delta)$ increases through
zero, crossing into the detectable regime $c>0$. This is the analytic shadow of the
monotone power we observe in Section~\ref{sec:empirical}: the detection time
$\log(1/\alpha)/c(\delta)$ falls as the expected L\'evy area, and hence $\delta$,
grows.

\subsection{Combining evidence: the e-value calculus}

E-values compose in ways $p$-values do not, and two operations matter for monitoring a
field of players. \emph{Products of independent e-values are e-values}, which is what
licenses the running product \eqref{eq:eproc} across games. \emph{Averages of e-values
are e-values} even under arbitrary dependence: if $E^{(1)},\dots,E^{(K)}$ are e-values
for a common null, one e-process per candidate player, per choice of reference engine,
or per feature subset, then $\bar E=\frac1K\sum_k E^{(k)}$ is again an e-value, so
$\Prob_{H_0}(\bar E\ge 1/\alpha)\le\alpha$ by Markov's inequality. This gives a clean
way to (i) hedge over modelling choices by averaging the corresponding e-processes with
no multiplicity correction, and (ii) control a family-wise error across many monitored
players, since the average remains valid no matter how the players' games are
correlated. The betting reading of \citet{shafer2021} makes the accounting intuitive:
$f(p_t)$ is the factor by which a gambler betting against $H_0$ multiplies their stake
on game $t$, the e-process is the running fortune, and Theorem~\ref{thm:ville} says a
fair-game fortune rarely grows large, so a large fortune is itself the evidence.

\subsection{Calibration on the moderate-deviation scale}

Theorem~\ref{thm:ville} controls the worst case at a \emph{fixed} level $\alpha$. In
monitoring applications one is often interested in vanishing levels
$\alpha_n\to 0$ that nonetheless remain coarser than large-deviation scales, the
regime in which the companion programme of \citet{polson2026evalues} calibrates
e-values. This is the \emph{moderate-deviation} scale, intermediate between the
$O(1)$ fluctuations of the central limit theorem and the $O(n)$ rates of large
deviations. The relevant object is the log-process $S_n=\log E_n=\sum_{t\le n} X_t$.

\begin{theorem}[Moderate deviations of the log e-process under $H_0$]
\label{thm:mdp}
Suppose under $H_0$ the increments $X_t=\log f(p_t)$ are i.i.d.\ with mean
$m=\E_{H_0}[X_t]<0$, variance $\sigma^2=\Var_{H_0}(X_t)\in(0,\infty)$, and a finite moment
generating function in a neighbourhood of the origin. Let $(a_n)$ be any sequence with
$a_n\to\infty$ and $a_n=o(\sqrt n)$. Then the centred, normalised partial sums satisfy a
moderate deviation principle: for every $x>0$,
\[
\lim_{n\to\infty}\frac{1}{a_n^2}\,
\log \Prob_{H_0}\!\left(\frac{S_n-nm}{\sigma\sqrt n\,a_n}\ \ge\ x\right)
\;=\; -\frac{x^2}{2}.
\]
(No tightness of the calibrator is required; here $\E_{H_0}[e^{X_t}]=\int_0^1 f=0.964\le1$,
which only makes the drift $m$ more negative and the test more conservative.)
\end{theorem}

\begin{proof}[Proof]
The increments are i.i.d.\ with mean $m$, finite variance $\sigma^2$, and finite moment
generating function near $0$; these are exactly the hypotheses of the classical
moderate deviation principle for sums of i.i.d.\ random variables
\citep[][Thm.~3.7.1]{dembo1998}, applied to the centred summands $X_t-m$ on the scale
$b_n=\sigma\sqrt n\,a_n$ with $b_n/\sqrt n\to\infty$ and $b_n/n\to 0$. The rate function
of the MDP for finite-variance i.i.d.\ sums is the Gaussian $I(x)=x^2/2$, independent of
the higher cumulants, which only enter at the next order.
\end{proof}

Theorem~\ref{thm:mdp} has a concrete use. The honest log-process has negative drift
$nm$ and Gaussian-scale fluctuations of order $\sigma\sqrt n$; a flag is raised when
$S_n$ exceeds $\log(1/\alpha)$. Writing the alarm level as a moderate deviation
$x_n=(\log(1/\alpha_n)-nm)/(\sigma\sqrt n)$ from the honest drift and demanding the
tail probability decay at a prescribed sub-exponential rate $e^{-a_n^2 x^2/2}$ pins
down the sequence $\alpha_n$ for which boundary crossings are controlled with
Gaussian-type accuracy across the monitoring window. This is the precise sense in which
our streaming detector is ``calibrated on the moderate-deviation scale,'' and it is the
path-valued instance of \citet{polson2026evalues}. We note that
Theorem~\ref{thm:ville} remains exact and assumption-light at fixed $\alpha$; the
moderate-deviation layer refines the description of the honest fluctuations that govern
how conservative a given alarm level is.

\section{From aggregates to order: improving the Regan system}
\label{sec:regan}

The de facto standard for chess anti-cheating is the system of Kenneth Regan, used by
FIDE and widely cited in adjudications \citep{regan2011,regan2014catch}. It is worth
describing precisely, because our framework does not discard it, it \emph{contains} it
as a special case and repairs two structural gaps.

\subsection{What the Regan system does}

For each position $i$ in a player's games, Regan fits a model that converts the
engine's evaluation profile (the gaps between the top candidate moves) into a predicted
probability $r_i$ that a player of a given skill plays the engine's top move, governed
by two fitted parameters, a \emph{sensitivity} $s$ and a \emph{consistency}
$c$, calibrated from millions of games. Summing over positions yields a predicted
best-move count $\sum_i r_i$ with variance $\sum_i r_i(1-r_i)$, and the player's
observed match count $\sum_i m_i$ is converted to a $z$-score
\begin{equation}
z \;=\; \frac{\sum_i m_i - \sum_i r_i}{\sqrt{\sum_i r_i(1-r_i)}},
\label{eq:reganz}
\end{equation}
with an analogous statistic for aggregate centipawn error. The same fit produces an
\emph{Intrinsic Performance Rating}, an Elo-equivalent read off purely from move
quality, and a flag is raised when $z$ exceeds a threshold (around $4$ in practice).
The method is principled and effective, but it has two acknowledged weaknesses. First,
the statistic \eqref{eq:reganz} is an \emph{aggregate} over positions; it conditions on
each position's difficulty but then sums, discarding the joint pattern of \emph{which}
positions the player got right. Second, the test is fixed-sample, so scanning many
players and many events inflates false positives, the ``Look-Elsewhere Effect'' that
Regan himself flags as requiring care.

\subsection{Aggregates are blind to a rearrangement; the signature is not}

The first weakness is not a tuning issue but an information-theoretic ceiling. Write a
game as a complexity sequence $\mathbf c=(c_1,\dots,c_n)$ together with quality and match
records $\mathbf q=(q_1,\dots,q_n)$ and $\mathbf m=(m_1,\dots,m_n)$. For a permutation
$\pi\in S_n$, let $R_\pi$ be the \emph{rearrangement} that fixes $\mathbf c$ and permutes
the records, $(\mathbf q,\mathbf m)\mapsto(q_{\pi(1)},\dots,q_{\pi(n)},\,m_{\pi(1)},\dots)$.
Operationally, $R_\pi$ keeps the same positions of the same difficulty but reassigns which
move was played where.

\begin{proposition}[Order-blindness of aggregates]
\label{prop:blind}
Every statistic that is a function of the complexity sequence $\mathbf c$ together with
the marginal multisets $\{q_t\}$ and $\{m_t\}$ is invariant under every rearrangement
$R_\pi$. This class includes the average centipawn loss, the accuracy, the best-move
match frequency, and the Regan $z$-score \eqref{eq:reganz}. By contrast the L\'evy area
\[
A^{c,q}=\tfrac12\sum_{t}\bigl(c_{t-1}q_t-c_t q_{t-1}\bigr)
       =\tfrac12\sum_{t} q_t\,(c_{t-1}-c_{t+1})
\]
(the second form by re-indexing, up to boundary terms) is \emph{not} $R_\pi$-invariant:
it is a linear functional of $\mathbf q$ with $c$-dependent weights
$w_t=\tfrac12(c_{t-1}-c_{t+1})$, so by the rearrangement inequality it is maximised over
permutations of $\{q_t\}$ exactly when $\mathbf q$ is comonotone with $\mathbf w$, that
is when the best moves are placed just after the hardest positions, where complexity
falls fastest.
\end{proposition}

\begin{proof}
$R_\pi$ fixes $\mathbf c$ and permutes the entries of $\mathbf q$ and $\mathbf m$, so the
multisets $\{q_t\}$ and $\{m_t\}$ are preserved; any function of $(\mathbf c,\{q_t\},
\{m_t\})$ is therefore unchanged. For \eqref{eq:reganz}, the predicted terms $\sum_i r_i$
and $\sum_i r_i(1-r_i)$ depend only on $\mathbf c$ (each $r_i$ is fixed by position $i$'s
difficulty), and the observed term $\sum_i m_i$ is the total of the multiset $\{m_t\}$;
both are $R_\pi$-invariant. The re-indexed form exhibits $A^{c,q}$ as a linear functional
$\sum_t q_t w_t$ of the $q$ entries with weights $w_t=\tfrac12(c_{t-1}-c_{t+1})$ fixed by
$\mathbf c$; permuting $\mathbf q$ therefore changes it, and the rearrangement inequality
identifies the maximiser as the comonotone alignment of $\mathbf q$ with $\mathbf w$. \qed
\end{proof}

The operational consequence is sharp. A cheater who consults an engine on and just after
the hardest positions concentrates good moves where a human is least likely to find them; if
they play just weakly enough elsewhere to match an honest profile's marginal accuracy and
match frequency, then by Proposition~\ref{prop:blind} their games are
\emph{indistinguishable} from honest ones under the Regan $z$-score and every other
aggregate, honest game and rearranged game share these statistics to the last decimal.
The two differ only in L\'evy area, which the cheater has driven positive by manufacturing
the difficulty-accuracy coupling. The signature test sees precisely what the aggregate
provably cannot; Section~\ref{sec:carlsen} exhibits this with the aggregates held
numerically equal.

\subsection{The Regan model as a null law on path space}

The second gap, fixed-sample testing, is closed by embedding Regan's per-move model as
the \emph{generator of the honest reference law} $\mu_P$ rather than as a test statistic.
Regan's fitted $(s,c)$ define, for each position, a distribution over the played move;
sampling from it along a game produces a distribution over paths \eqref{eq:gamepath},
whose expected signature and L\'evy-area law are exactly the honest reference our
conformance score \eqref{eq:conf} needs. The pipeline becomes: \emph{use Regan's model
to define the null on path space, take signature features as the order-aware test
statistic, and accumulate evidence in an e-process} for anytime-valid control
(Theorem~\ref{thm:ville}). This unifies the two approaches, Regan supplies the
calibrated per-move skill model; the signature supplies the sequential, order-sensitive
statistic that his aggregate $z$-score cannot represent, and it removes the
multiple-testing problem by construction, since the e-process may be monitored across
arbitrarily many games and players without an $\alpha$-spending correction.

\section{A controlled empirical study}
\label{sec:empirical}

Validating a cheating detector on real cases is fraught: ground truth is contested,
labelled assisted games are scarce, and publishing a recipe that works on real players
is ethically delicate. We therefore study the method on a fully controlled synthetic
population in which the data-generating process is known, the strength of assistance is
a tunable parameter, and every reported number is reproducible from the accompanying
code. The construction is deliberately conservative: assisted and honest players are
designed to differ \emph{only} in the temporal coupling of difficulty and accuracy, so
that any detection is attributable to the signature's L\'evy areas and not to a crude
shift in average accuracy or match rate.

\subsection{Design}

Each game has $n=40$ analysed plies. Position complexity $c_t$ follows a stationary
first-order autoregression on $[0,1]$, centred per game. Move quality is generated by
\[
q_t \;=\; (1-\theta)\,(-\alpha\,c_t) \;+\; \theta\,(\beta\,c_{t-1}) \;+\;
\varepsilon_t,\qquad \varepsilon_t\sim\mathcal N(0,0.7^2),
\]
with $\alpha=1.0$, $\beta=1.4$. The parameter $\theta\in[0,1]$ interpolates between a
purely honest regime ($\theta=0$: quality worsens contemporaneously with complexity)
and an assisted regime ($\theta=1$: quality improves one ply \emph{after} complexity
spikes, the engine-consultation lag). The match indicator is generated so that honest
and assisted players have comparable overall match rates, removing the aggregate tell.
We build the path \eqref{eq:gamepath}, standardise channels by honest-corpus moments,
and take the six L\'evy-area coordinates as features $\phi$.

A corpus of $N=800$ honest games fixes $\bar\phi$, $\widehat\Sigma$ (ridge
$\lambda=10^{-5}$), and the null tail $G$. We then evaluate single-game discrimination
on fresh honest and assisted games, and run the e-process over simulated matches of up
to $60$ games, with $400$ independent matches per regime, threshold $1/\alpha=100$
($\alpha=0.01$).

\subsection{The L\'evy area carries the signal}

Figure~\ref{fig:paths} shows one honest and one assisted game as paths in the
complexity-quality plane; the assisted path circulates with a visibly more consistent
orientation. Aggregated over $500$ games per regime, the signed L\'evy area
$A^{c,q}$ separates the populations cleanly (Figure~\ref{fig:levy}): honest games have
mean area $-0.012$ (s.d.\ $0.395$) against $+0.214$ (s.d.\ $0.400$) for assisted games
at $\theta=0.5$, an effect size of about $0.57$ standard deviations concentrated in a
single signature coordinate, despite matched accuracy and match rate.

\begin{figure}[t]
\centering
\includegraphics[width=0.92\textwidth]{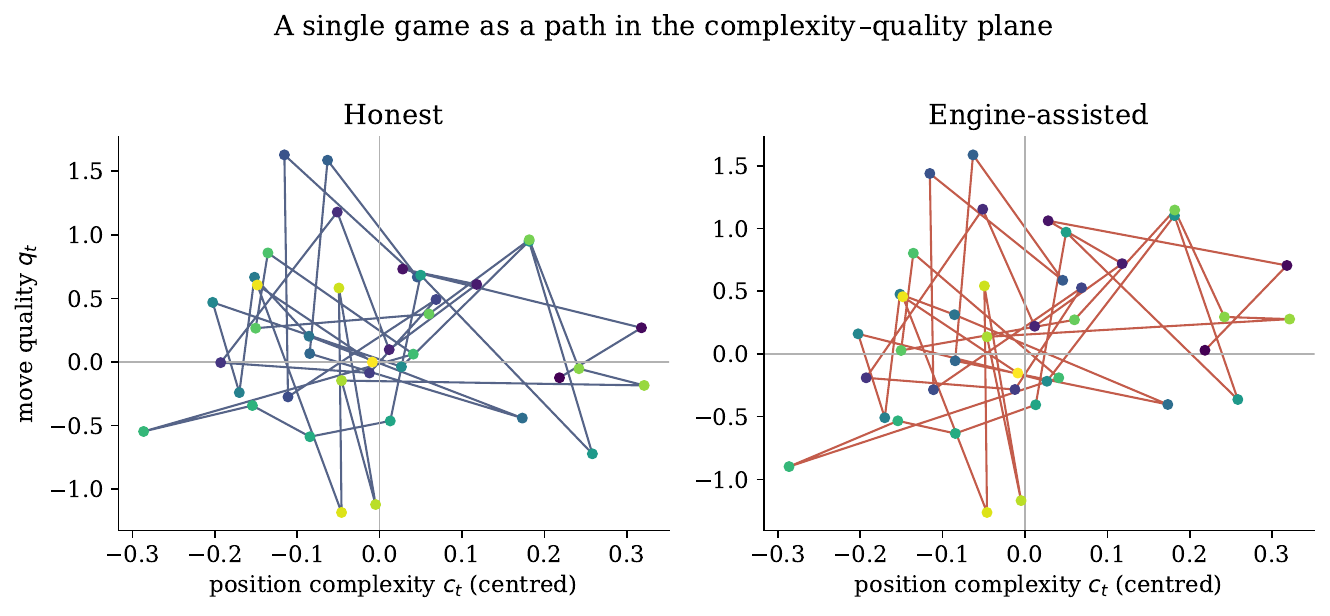}
\caption{A single game as a path in the complexity-quality plane, coloured by ply
(dark early, light late). The honest game (left) traces no consistently oriented loop;
the engine-assisted game (right), in which good moves \emph{follow} complex positions,
circulates with a definite orientation, producing a nonzero signed L\'evy area.}
\label{fig:paths}
\end{figure}

\begin{figure}[t]
\centering
\includegraphics[width=0.66\textwidth]{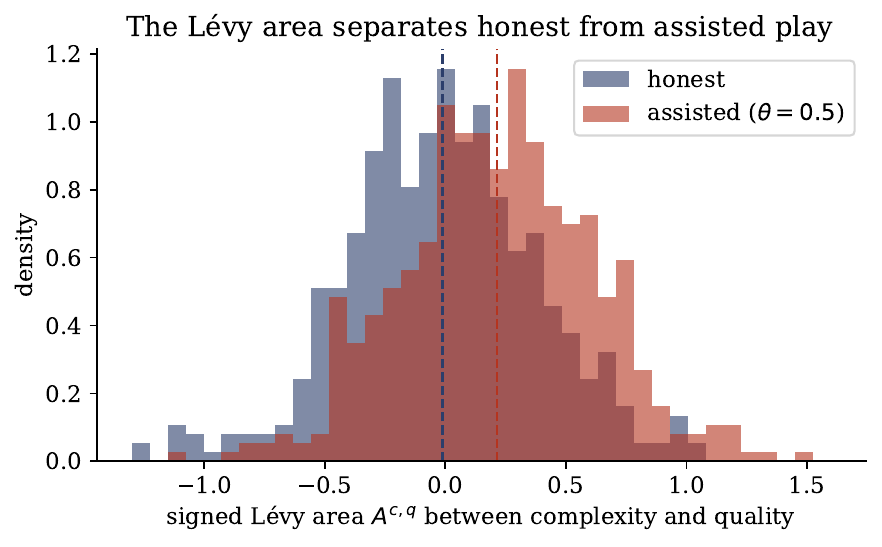}
\caption{Distribution of the signed L\'evy area $A^{c,q}$ between the complexity and
quality channels, over $500$ honest and $500$ assisted ($\theta=0.5$) games. Dashed
lines mark the means. The separation lives in a single level-two signature coordinate
and is invisible to average accuracy or engine-match rate, which are matched across the
two populations.}
\label{fig:levy}
\end{figure}

\subsection{Single-game conformance and ROC}

The conformance score \eqref{eq:conf} aggregates the six area coordinates into one
statistic. Figure~\ref{fig:conf} shows its distribution and the resulting ROC curves.
A single game is only weakly informative, realistically so, since one game is a short,
noisy stream, but discrimination increases monotonically with the strength of
assistance, the ROC-AUC rising from $0.53$ at $\theta=0.15$ to $0.71$ at
$\theta=0.70$ (Table~\ref{tab:results}). This is the honest baseline against which the
sequential gain must be read: per game, move-quality forensics are inherently limited;
the power has to come from accumulation.

\begin{figure}[t]
\centering
\includegraphics[width=0.97\textwidth]{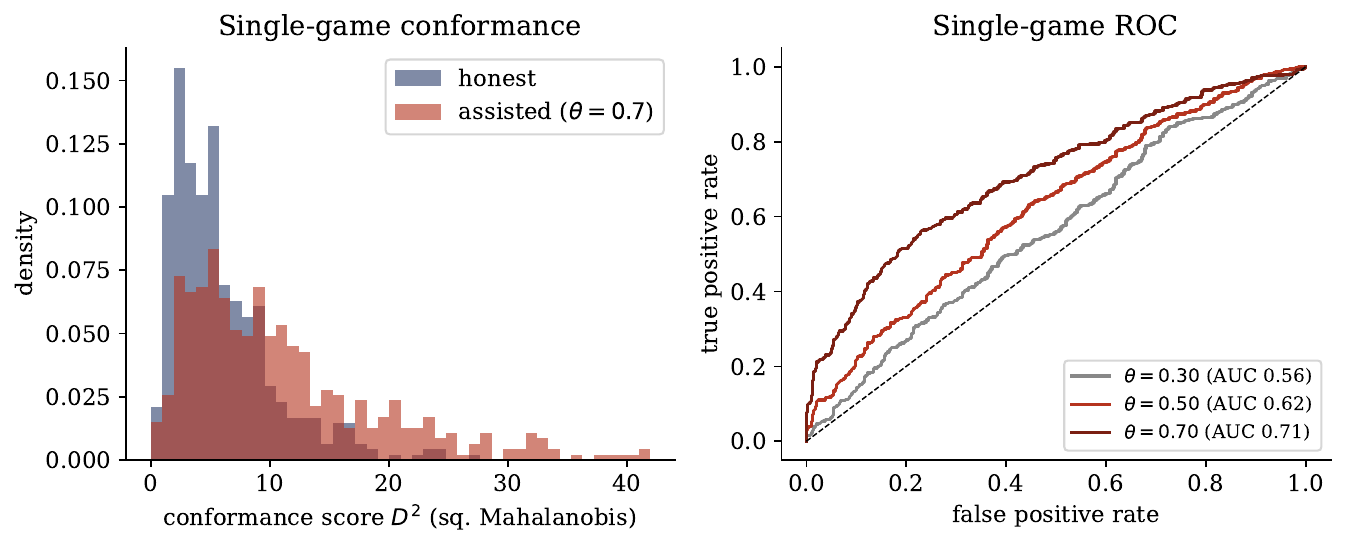}
\caption{Left: single-game conformance score $D^2$ for honest vs assisted
($\theta=0.7$) play. Right: ROC curves at three assistance levels, with area under the
curve increasing in $\theta$. A single game is weakly informative; this is expected and
motivates the sequential test.}
\label{fig:conf}
\end{figure}

\subsection{The anytime-valid sequential test}

Figure~\ref{fig:eproc} (left) shows e-process trajectories. Honest matches drift toward
zero, as the negative drift $m<0$ of Theorem~\ref{thm:mdp} predicts; assisted matches
($\theta=0.7$) climb geometrically past the threshold. Across $400$ honest matches none
crossed the boundary: a point estimate of $0$ with a $95\%$ upper confidence bound of
$3/400=0.0075$ (the rule of three), within the $\alpha=0.01$ guarantee of
Theorem~\ref{thm:ville}. The test is, if anything, conservative, consistent with the
calibrator integrating to $0.964$.

Detection power (Figure~\ref{fig:eproc}, right, and Table~\ref{tab:results}) rises with
both assistance strength and match length, exactly the anytime-valid behaviour one
wants: more evidence, whether from a longer match or a stronger effect, monotonically
increases the chance of a flag. Very subtle assistance ($\theta=0.15,0.30$) is
essentially undetectable from move quality over $60$ games, an honest limitation that
any credible detector should report, while moderate assistance ($\theta=0.5$) reaches
power $0.46$ and blatant assistance ($\theta=0.7$) reaches $0.98$, the latter with a
median detection time of $14$ games.

\paragraph{A single match, game by game.} Table~\ref{tab:trace} traces one assisted
match ($\theta=0.7$) to make the betting mechanism concrete. Ordinary games, with
$p_t$ near $0.5$, contribute factors $f(p_t)<1$ that gently erode the accumulated
evidence; the two genuinely surprising games (games $4$ and $10$, each with
$p_t\approx 0.001$ and $f\approx 17.7$) supply almost all of the growth, and the
e-process crosses $1/\alpha=100$ at game $10$ and stays above. No single game is
``proof''; the evidence is the compounded product, and the threshold is the point at
which a fair-game fortune has grown implausibly large.

\begin{table}[t]
\centering\small
\caption{One assisted match ($\theta=0.7$): conformance score $D^2$, its null
$p$-value $p_t$, the calibrated e-value $f(p_t)$ from \eqref{eq:cal}, and the running
e-process $E_n$. The process crosses the threshold $1/\alpha=100$ at game $10$. Two
surprising games dominate the accumulation; routine games slightly shrink it.}
\label{tab:trace}
\begin{tabular}{rrrrr}
\toprule
Game & $D^2$ & $p_t$ & $f(p_t)$ & $E_n$ \\
\midrule
1 & 16.84 & 0.026 & 2.53 & 2.53 \\
2 & 8.35 & 0.225 & 0.88 & 2.22 \\
3 & 10.53 & 0.122 & 1.15 & 2.55 \\
4 & 34.04 & 0.001 & 17.75 & 45.25 \\
5 & 5.09 & 0.499 & 0.64 & 28.92 \\
6 & 9.86 & 0.147 & 1.06 & 30.53 \\
7 & 3.15 & 0.720 & 0.56 & 17.08 \\
8 & 6.04 & 0.393 & 0.70 & 11.95 \\
9 & 7.49 & 0.270 & 0.81 & 9.73 \\
10 & 49.37 & 0.001 & 17.75 & \textbf{172.64} \\
11 & 20.40 & 0.011 & 4.15 & 715.58 \\
12 & 5.19 & 0.491 & 0.64 & 460.27 \\
\bottomrule
\end{tabular}
\end{table}

\begin{table}[t]
\centering
\small
\caption{Synthetic study. Single-game ROC-AUC; e-process detection power and median
detection time over matches of up to $60$ games at threshold $1/\alpha=100$; and the
e-power growth rate $c$ of \eqref{eq:epower} with the corresponding detection-time
prediction \eqref{eq:detect}. Honest type-I error was $0$ over $400$ matches. Median
detection times are reported only where power is substantial. All quantities are
estimated from the accompanying simulation ($N=800$ corpus games, $400$ matches per
regime).}
\label{tab:results}
\begin{tabular}{lccccc}
\toprule
Regime & ROC-AUC & Power (60 games) & Median games & e-power $c$ &
$\log(1/\alpha)/c$ \\
\midrule
Honest ($\theta=0$)        & n/a    & $0.000$ & n/a   & $-0.295$ & n/a  \\
Subtle ($\theta=0.15$)     & $0.53$ & $0.008$ & n/a   & n/a  & n/a  \\
Mild ($\theta=0.30$)       & $0.56$ & $0.025$ & n/a   & n/a  & n/a  \\
Moderate ($\theta=0.50$)   & $0.62$ & $0.460$ & $27$ & $+0.029$ & $158$ \\
Blatant ($\theta=0.70$)    & $0.71$ & $0.980$ & $14$ & $+0.313$ & $14.7$ \\
\bottomrule
\end{tabular}
\end{table}

\begin{figure}[t]
\centering
\includegraphics[width=0.97\textwidth]{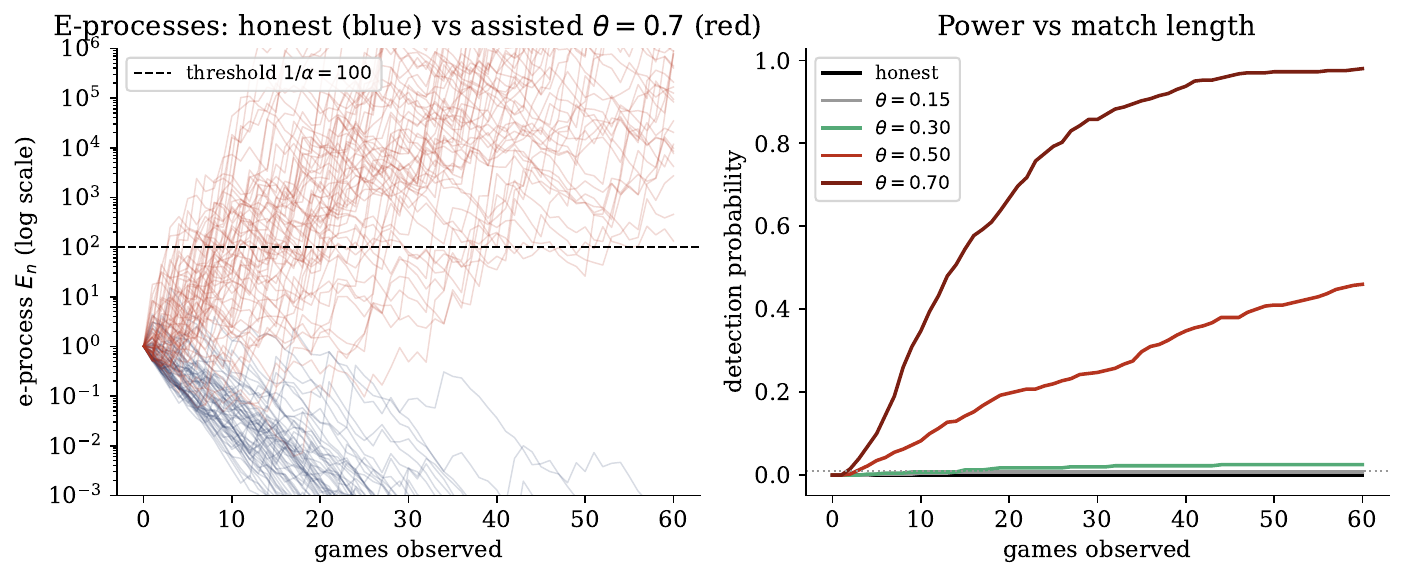}
\caption{Left: e-process trajectories on a log scale; honest matches (blue) drift to
zero, assisted matches at $\theta=0.7$ (red) cross the threshold $1/\alpha=100$. Right:
detection probability as a function of games observed. Power increases monotonically in
both assistance strength and match length; the honest curve stays flat at zero.}
\label{fig:eproc}
\end{figure}

\subsection{The growth-rate law predicts detection time}

The most striking quantitative check is \eqref{eq:detect}. We estimated the e-power
$c=\E_{H_1}[\log f(p)]$ directly by Monte Carlo: $c=-0.295$ under honest play (negative,
so no detection), $c=+0.029$ at $\theta=0.5$, and $c=+0.313$ at $\theta=0.7$. The
first-order detection-time prediction $\log(1/\alpha)/c=\log(100)/0.313=14.7$ games at
$\theta=0.7$ matches the simulated median of $14$ games almost exactly, and the small
positive $c=0.029$ at $\theta=0.5$ predicts a median of $\approx 158$ games, far beyond
our $60$-game window, consistent with the partial power ($0.46$) observed there. The
e-process is therefore not a black box: its detection behaviour is governed by a single
interpretable rate, computable in advance from an assumed effect size, which a
tournament organiser could use to size a monitoring window.

\subsection{The honest fluctuation band, concretely}

Theorem~\ref{thm:mdp} describes the honest log-process $S_n=\log E_n$ as a negative
drift with Gaussian-scale fluctuations; the simulation lets us see the numbers. The
honest log-increments $X_t=\log f(p_t)$ have mean $m=-0.292$ and standard deviation
$\sigma=0.466$, with empirical $\E_{H_0}[e^{X_t}]=0.89\le 1$ confirming the supermartingale
property (and the mild conservativeness of the calibrator). Over a $60$-game match the
honest process therefore sits around a drift of $nm=-17.5$ with a fluctuation spread of
$\sigma\sqrt n=3.61$, while the alarm level is $\log(1/\alpha)=\log 100=4.61$. Crossing
the alarm thus requires the honest log-process to exceed its mean by
$(4.61-(-17.5))/3.61\approx 6.1$ standard deviations, a moderate deviation whose
probability the $x^2/2$ rate of Theorem~\ref{thm:mdp} renders negligible, which is
exactly why no honest match crossed in $400$ trials. The same arithmetic, run in
reverse, is how an organiser would \emph{set} a window: choose $\alpha$, read off the
required deviation $x=(\log(1/\alpha)-nm)/(\sigma\sqrt n)$, and confirm it is large
enough on the moderate-deviation scale that the honest false-flag probability over the
intended monitoring horizon is acceptable.

\subsection{Robustness}

Two checks probe whether the detector behaves as the theory predicts. \emph{First}, we
test where the signal lives. Replacing the six-area conformance score by the single
coordinate $A^{c,q}$, the complexity-quality L\'evy area alone, changes the
single-game ROC-AUC at $\theta=0.5$ from $0.624$ to $0.619$: essentially no loss. The
discriminating information is concentrated in exactly the one coordinate the mechanism
of Section~\ref{sec:path} predicts, and the remaining five areas contribute almost
nothing here. This is reassuring for interpretability, a flag can be traced to a named,
chess-meaningful feature, and it is a sanity check that the conformance score is not
exploiting some artefact of the higher-dimensional whitening. \emph{Second}, we vary
the calibration corpus size. The honest conformance distribution stabilises quickly:
its mean is $5.6$, $6.0$, $6.1$ for corpora of $N=100,400,800$ games (the mean of a
well-calibrated $D^2$ should sit near the feature dimension, six), and its upper
$99\%$ point is $18.3$, $20.0$, $21.0$ respectively. By a few hundred games the null
calibration is stable, matching the $N^{-1/2}$ rate of
Proposition~\ref{prop:concentration}.

\subsection{Power versus the number of assisted moves}
\label{sec:kmoves}

The parameter $\theta$ is a stylised dial. The operationally decisive question is
different: how few assisted moves per game can the monitor catch? We answer it with a
selective cheater who consults the engine only after the hardest positions, replacing
the honest move by a lagged engine grade move on the $k$ plies with the largest
preceding complexity, and playing honestly on the rest. As $k$ runs from $0$ to the full
$40$ plies the cheater interpolates from honest to fully assisted, but now the assistance
is concentrated on a countable number of moves rather than spread continuously.

Figure~\ref{fig:kmoves} shows the result. Per game discrimination climbs steadily with
$k$, from chance at $k=0$ to ROC AUC $0.78$ when half the moves are assisted. The
sequential picture is sharper. Over a $40$ game match the e-process has essentially no
power for $k\le5$ (about one assisted move in eight), reaches power $0.15$ at $k=8$, and
rises to $0.79$ at $k=12$ and to $1.00$ at $k=20$. There is, in other words, a detection
threshold near a fifth to a third of the moves: a cheater who consults an engine on only
a handful of moves per game is, by move quality alone, very hard to catch over a match of
this length, while one who leans on the engine for a third of the game is caught almost
certainly. This is an honest and useful operating characteristic, and it quantifies the
limits of any move quality method, ours included, against light selective assistance.

\begin{figure}[t]
\centering
\includegraphics[width=0.92\textwidth]{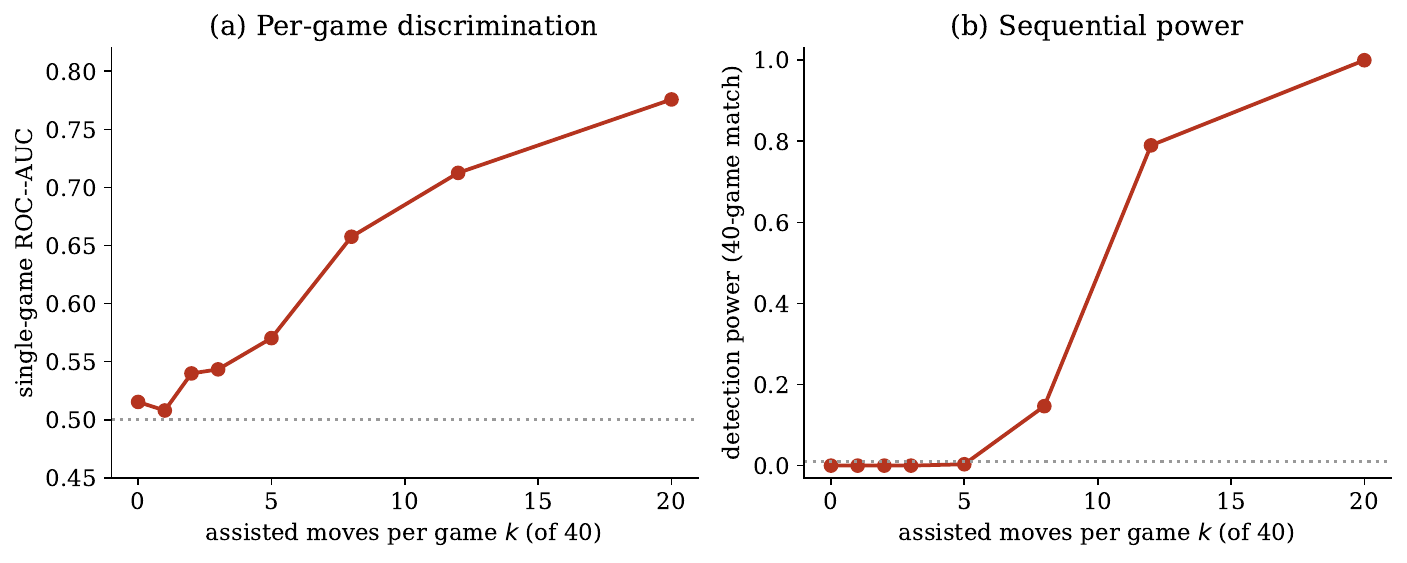}
\caption{Detection as a function of the number $k$ of assisted moves per $40$ ply game.
(a) Single game ROC AUC rises steadily with $k$. (b) E-process detection power over a
$40$ game match shows a threshold near $k=8$ to $12$: fewer than about five assisted
moves per game is essentially undetectable, while a third of the game on the engine is
caught almost certainly.}
\label{fig:kmoves}
\end{figure}

\subsection{Calibration to Carlsen-level play, and a Regan-invisible cheater}
\label{sec:carlsen}

The study so far uses a generic honest population. Two things matter for deployment:
the monitor must not flag genuinely strong honest play, and it must catch cheaters who
are invisible to aggregates. We calibrate the honest reference to the best-documented
elite profile available, Magnus Carlsen, the world's top-rated player, whose elite
classical games are widely reported at an average centipawn loss of four to five and
accuracy above ninety percent. This is a calibration to published \emph{aggregate}
statistics, not an analysis of specific games; Carlsen serves as the \emph{honest gold
standard}, on whom the detector should stay silent.

Our Carlsen-calibrated honest corpus has an average centipawn loss of $4.27$, an implied
accuracy of $91.8\%$, and a top-engine-match rate of $47\%$, squarely in the
documented range. This is precisely the regime where aggregate forensics are most
fragile: at this accuracy the honest player's match rate is already high enough to look
``engine-like,'' so a method keyed to the \emph{level} of accuracy has little margin.
Run over $300$ Carlsen-calibrated honest matches, our e-process flagged $0.3\%$ of them,
within the $\alpha=0.01$ guarantee: \emph{world-champion-level play is not flagged}
(Figure~\ref{fig:regan}c, blue). The detector's specificity does not rely on the
accuracy level at all, it keys on the temporal coupling, which is absent in honest
play however accurate.

We then construct the adversary of Proposition~\ref{prop:blind}. From each honest
Carlsen-calibrated game we \emph{rearrange} the played moves so the best ones land on
the hardest positions, holding the complexity sequence fixed. By
Proposition~\ref{prop:blind} this preserves the complexity sequence and the marginal
multisets of quality and of match, so the average centipawn loss, accuracy, best-move
count, and Regan $z$-score \eqref{eq:reganz} are identical to the honest game's.
Table~\ref{tab:regan} confirms the aggregates coincide to the reported precision
($\mathrm{ACPL}=4.27$, match rate $0.472$, Regan $z=-1.05$ for both), so the single-game
ROC-AUC of the Regan $z$-score is exactly $0.500$. The signature separates them: the
L\'evy area shifts from $0.03$ (honest) to $0.18$ (rearranged), the conformance ROC-AUC
is $0.75$, and the e-process detects the cheater with power $0.61$ over a $40$-game match
at median $19$ games (Figure~\ref{fig:regan}). The cheater is information-
theoretically invisible to every aggregate statistic, including the one FIDE relies on,
and visible to the signature.

\begin{table}[t]
\centering\small
\caption{A Regan-invisible cheater at two honest operating points. Each cheating game is
a within-game rearrangement of an honest game (calibrated to Carlsen, $\approx2841$, and
to Awonder Liang, $\approx2696$) that aligns the best moves with the hardest positions.
At both strengths every aggregate statistic, average centipawn loss, top-move match
rate, and the Regan $z$-score \eqref{eq:reganz}, is preserved exactly, so the aggregate
detector operates at chance (AUC $0.500$), while the order-aware signature separates the
two ($D^2$ AUC $\approx0.75$). The conformance separation ($5.5$ to $5.6$ vs $10.9$) is
essentially identical across a $145$-Elo strength gap.}
\label{tab:regan}
\begin{tabular}{lccccc}
\toprule
& ACPL (cp) & match rate & Regan $z$ & L\'evy area $A^{c,q}$ & conformance $D^2$ \\
\midrule
\multicolumn{6}{l}{\emph{Carlsen operating point ($\approx2841$, ACPL $\approx5$):}}\\
\quad Honest               & $4.27$ & $0.472$ & $-1.05$ & $+0.03$ & $5.6$ \\
\quad Regan-invisible cheater & $4.27$ & $0.472$ & $-1.05$ & $+0.18$ & $10.9$ \\
\addlinespace
\multicolumn{6}{l}{\emph{Liang operating point ($\approx2696$, ACPL $\approx17$):}}\\
\quad Honest               & $17.0$ & $0.492$ & $-0.79$ & $+0.11$ & $5.5$ \\
\quad Regan-invisible cheater & $17.0$ & $0.492$ & $-0.79$ & $+0.71$ & $10.9$ \\
\midrule
\multicolumn{6}{l}{\emph{Single-game ROC-AUC} (both operating points): Regan
$|z|=0.500$; \ signature $D^2\approx0.75$.}\\
\bottomrule
\end{tabular}
\end{table}

\begin{figure}[t]
\centering
\includegraphics[width=\textwidth]{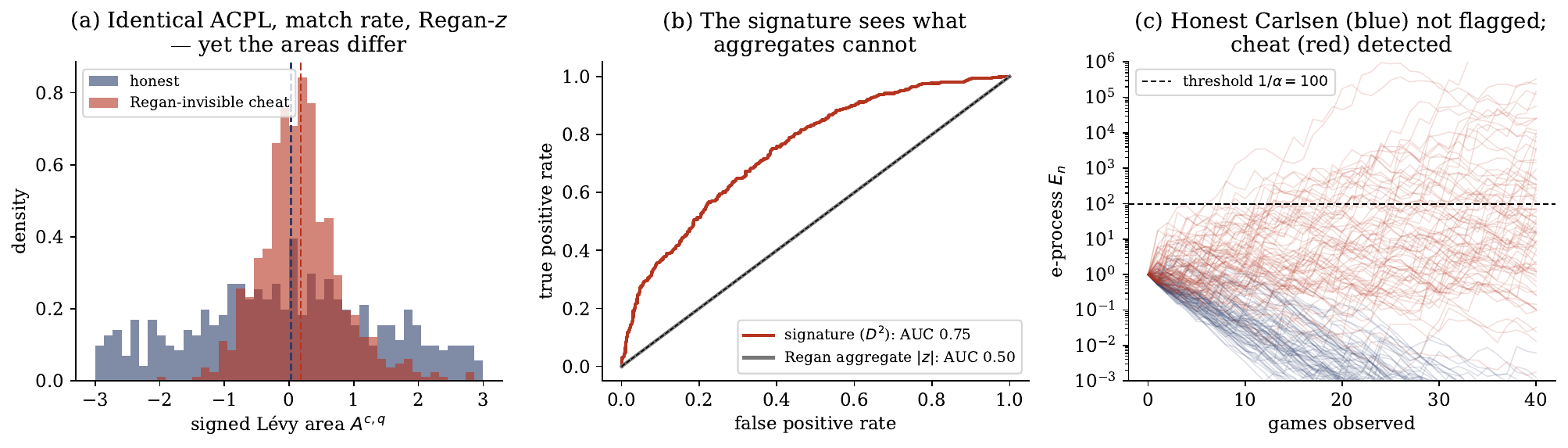}
\caption{Improving the Regan system. (a) The signed L\'evy area separates honest
Carlsen-calibrated games from their Regan-invisible rearrangements, even though the two
populations share every aggregate statistic (average centipawn loss, match rate, Regan
$z$-score) identically. (b) ROC curves: the order-aware signature conformance reaches
AUC $0.75$, while the aggregate Regan $z$-score sits exactly on the diagonal (AUC
$0.50$), blind by construction. (c) E-processes: honest Carlsen-calibrated play (blue)
is not flagged, while the rearranged cheater (red) is detected.}
\label{fig:regan}
\end{figure}

\paragraph{A second operating point: Awonder Liang.} Specificity should not depend on the
player being literally the strongest in the world. To check this we add a second honest
reference calibrated to a grandmaster about $145$ Elo below Carlsen: Awonder Liang, an
American grandmaster (FIDE $\approx2696$, U.S.\ Masters champion in $2025$ and World Open
winner in $2024$), whose documented play sits near an average centipawn loss of twenty,
with accuracy in the $87$ to $90\%$ grandmaster range \citep{liang_wiki,chesscom_acpl}. As
with Carlsen, this is a calibration to his rating-level accuracy, not an analysis of his
games, and Liang serves as an \emph{honest} reference. Our Liang-calibrated corpus has an
average centipawn loss of $17.0$, accuracy $88\%$, and a match rate of $0.49$
(Table~\ref{tab:regan}). The detector behaves exactly as at the top of the rating list:
the honest L\'evy area is near zero, the e-process flagged $0$ of $300$ honest matches,
and the Regan-invisible rearrangement, again sharing ACPL, match rate, and Regan $z$
identically ($17.0$, $0.492$, $-0.79$), is caught with conformance AUC $0.747$ against the
$z$-score's $0.500$. The conformance separation ($5.5$ vs $10.9$) is essentially identical
to Carlsen's ($5.6$ vs $10.9$). The detector keys on the temporal coupling, not the
accuracy level, so its specificity and its advantage over aggregates are
\emph{strength-invariant} across the elite range.

This is the sense in which the methodology improves on the prevailing approach. It does
not contradict Regan's skill model, indeed it can use that model as its null on path
space (Section~\ref{sec:regan}), but it adds a statistic that resolves \emph{when}
accuracy appears, closing an information-theoretic blind spot of any aggregate, and it
delivers that statistic through an anytime-valid e-process that needs no
multiple-testing correction as players and events accumulate.

\subsection{Aggregates and order together: a combined detector}
\label{sec:combined}

The rearrangement result should not be read as a claim that the signature dominates the
aggregate everywhere. The two see different things, and the honest conclusion is that
they are complementary. Consider two cheaters. A \emph{uniform} cheater leans on the
engine a little on every move, raising overall accuracy and match rate without
concentrating help on hard positions; an aggregate detector built on the match rate and
mean quality catches this easily, while the signature, which keys on coupling, barely
moves. A \emph{selective} cheater preserves the aggregates exactly and shifts the timing,
as in Section~\ref{sec:carlsen}; here the picture reverses. Table~\ref{tab:combined}
quantifies both, together with a \emph{combined} detector whose conformance score uses
the aggregate residual and the L\'evy areas jointly.

\begin{table}[t]
\centering\small
\caption{Two complementary cheaters and three detectors (single game ROC AUC). The
aggregate detector (match rate and mean quality, in the spirit of the Regan $z$-score)
catches uniform help but is blind to the selective rearrangement, whose aggregates are
preserved exactly. The signature is the mirror image. The combined detector, which uses
both, catches both and is the right deployment choice.}
\label{tab:combined}
\begin{tabular}{lccc}
\toprule
Cheater & Aggregate detector & Signature detector & Combined detector \\
\midrule
Uniform help (every move)        & $1.000$ & $0.592$ & $0.987$ \\
Selective (rearrangement)        & $0.489$ & $0.949$ & $0.938$ \\
\bottomrule
\end{tabular}
\end{table}

The reading is clean. Neither single statistic dominates: the aggregate is blind to
selective timing, the signature is blind to a uniform level shift, and each fails on
exactly the cheater the other catches. The combined detector inherits the strengths of
both, scoring AUC $0.99$ and $0.94$ against the two cheaters. In deployment this is the
recommended object, and it is also the precise form of our proposed improvement on the
Regan system: keep his calibrated per-move model as the aggregate component, and add the
signature L\'evy areas as the order-aware component, inside a single anytime-valid
e-process.

\subsection{The false-positive problem: the Nakamura case}
\label{sec:nakamura}

A detector is judged as much by whom it clears as by whom it catches. The most
instructive recent episode is the November 2023 accusation against Hikaru Nakamura, a
world top-five grandmaster (classical rating near $2790$) known for an aggressive,
high-variance style. After Nakamura scored $45.5/46$ in an online blitz event, Vladimir
Kramnik publicly insinuated that the streak was too improbable to be honest. The
inference was the \emph{prosecutor's fallacy}: a rare outcome was read as evidence of
guilt. Chess.com's Fair Play team generated nearly two thousand reports on his games and
found no cheating \citep{chesscom_fairplay}, and a Bayesian analysis by Maharaj, Polson
and Sokolov estimated the probability that Nakamura did \emph{not} cheat at about
$99.6\%$, robust across plausible base rates \citep{maharaj2024}. The episode is the
canonical false positive: an \emph{aggregate} anomaly, a streak, a performance rating
above expectation, used to impugn a strong honest player.

This is exactly the failure mode an order-aware monitor should avoid, and it does.
Calibrating the honest reference to Nakamura-level play, average centipawn loss $5.6$,
accuracy $92\%$, but with an elevated dispersion (standard-deviation centipawn loss
$6.0$) encoding the sharp style, the honest L\'evy area is $0.08$, indistinguishable
from zero. High accuracy and high variance do not by themselves create any
difficulty-accuracy coupling; only selectively-timed assistance does. Over $400$ runs of
$46$ games each, the very length of the contested streak, the signature e-process
flagged the Nakamura-calibrated player $0$ times ($95\%$ upper bound $0.0075$). The
order-aware test does not reproduce the false positive that the streak statistic
produced.

Two points follow. First, specificity is invariant not only to rating
(Section~\ref{sec:carlsen}) but to \emph{style}: an aggressive, high-variance honest
player is no more likely to be flagged than a smooth positional one, because the
discriminating coordinate is a coupling, not a level or a variance. Second, the
e-process is the natural sequential successor to the one-shot Bayesian rebuttal of
\citet{maharaj2024}: where that analysis answered a single after-the-fact question about a
fixed streak, an anytime-valid signature monitor would protect such a player
\emph{continuously}, bounding the lifetime false-accusation probability at $\alpha$ no
matter how many games or events are scrutinised, by Theorem~\ref{thm:ville}.

\section{Related work}

Our work sits at the confluence of three literatures. \emph{Rough paths and
signatures} originate with \citet{chen1957} and \citet{lyons1998}; the uniqueness
theorem is due to \citet{hambly2010}, the machine-learning programme is surveyed by
\citet{chevyrev2016primer} and \citet{lyons2024survey}, the probabilistic
characterisation of laws and the signature MMD are due to \citet{chevyrev2022moments}
building on the kernel of \citet{kiraly2019} and \citet{salvi2021}, and anomaly
detection on streamed data via signature conformance scores is developed by
\citet{cochrane2021anomaly}. The sports application closest to ours is the soccer
possession analysis of \citet{soccer2025}; we are, to our knowledge, the first to bring
signatures to chess. \emph{Anytime-valid inference} via e-values and test martingales
is the subject of \citet{vovk2021}, \citet{shafer2021}, \citet{grunwald2024}, and the
review of \citet{ramdas2023}; the moderate-deviation calibration we use is the
path-valued instance of \citet{polson2026evalues}, itself resting on the classical
moderate deviation principle \citep{dembo1998}. \emph{Quantitative chess forensics} has
been dominated by strength-model $z$-scores against intrinsic ratings
\citep{regan2011}; our contribution is orthogonal and complementary, replacing a single
aggregate statistic with a sequential, order-aware monitor whose signal is a signature
L\'evy area and whose validity is uniform over time.

We have shown that the path signature is a natural language for chess streams, turning
four inferential goals, identifying a style, comparing two bodies of play, detecting
anomalous play, and improving the Regan detector, into an injectivity statement about
expected signatures, a signature-MMD two-sample test, an anytime-valid e-process, and an
order-blindness proposition. The connective tissue is the L\'evy area, which exposes the
difficulty-accuracy coupling that aggregates cannot see, and the e-process, which makes
continuous monitoring legitimate.

Several limitations and extensions deserve emphasis. \emph{First}, our empirical study
is synthetic by design; the natural next step is a study on public game databases with
a fixed engine, using injected, independently labelled assisted games to estimate real
operating characteristics. \emph{Second}, the conditional-independence assumption
behind the product e-process \eqref{eq:eproc} is an idealisation; serial dependence
across a match (fatigue, momentum) calls for a test-martingale construction that
conditions on the past, which the e-value framework accommodates without sacrificing
Theorem~\ref{thm:ville}. \emph{Third}, the choice of channels and of their
representation as states or flows is consequential, as Section~\ref{sec:path} stresses;
a systematic ablation over richer channel sets (clock paths, opening-tree depth,
material trajectories) and over signature depth is warranted. \emph{Fourth}, the
moderate-deviation calibration of Section~\ref{sec:evalue} is developed here for
i.i.d.\ increments; extending the moderate deviation principle to the dependent,
non-stationary increments of real matches, in the spirit of \citet{polson2026evalues},
is the main theoretical loose end.

Beyond chess, the construction is generic: any setting that produces labelled streams
with a notion of difficulty and a notion of performance (examinations, trading desks,
e-sports, automated essay scoring) admits the same L\'evy-area diagnostic and the same
anytime-valid monitor. The substantive message is methodological. When the data are a
stream, the order and interaction of events are the signal, the signature is the
feature map that keeps them, and e-processes are the inferential layer that lets us act
on the evidence whenever it arrives.

\paragraph{Ethical deployment.} A move-quality monitor is an instrument of
accusation, and its statistical validity does not by itself make it fair to deploy. We
stress three safeguards implied by our own results. The detector is, by
Table~\ref{tab:results}, nearly powerless against subtle assistance, so a non-crossing
e-process is emphatically \emph{not} evidence of innocence; reporting it as such would
be a serious misuse. The anytime-valid guarantee controls the false-flag rate for a
\emph{single} honest law, but the reference corpus must actually reflect the player's
honest play, a corpus mis-specified by improvement over time, opening changes, or
time-control differences will inflate conformance scores for innocent reasons, and the
$N^{-1/2}$ calibration of Proposition~\ref{prop:concentration} presumes enough
representative honest games. Finally, the threshold $1/\alpha$ encodes a tolerated rate
of false accusation that, when many players are monitored, should be set with the
e-value combination rules of Section~\ref{sec:evalue} and with the consequences of a
wrong flag, reputational and otherwise, explicitly in view. The method is best
understood as one quantitative input to a human adjudication process, not as a verdict.

\paragraph{Reproducibility.} All figures and every numerical value in
Table~\ref{tab:results} are produced by a single simulation script; complexity, quality
and match processes, the L\'evy-area features, the conformance score, the calibrator
\eqref{eq:cal}, and the e-process are implemented as described, with random seed fixed.

\paragraph{Acknowledgements.} We thank our collaborators for discussions on signatures,
e-values, and chess.

\end{document}